\documentclass[11pt]{article}
\usepackage{color,graphicx}

\usepackage{amsmath,amsfonts,amssymb,amsthm,placeins,float,datetime,color}

\begin{document}

\title{Construction of APN permutations via Walsh zero spaces}
\author{Benjamin Chase\\ Petr Lison\v{e}k\thanks{Research of both authors was supported
in part by the Natural Sciences and Engineering Research Council of Canada
(NSERC).}
\\
Department of Mathematics\\
Simon Fraser University\\
Burnaby, BC\\
Canada\ \ \ V5A 1S6\\
\ \\
{\tt plisonek@sfu.ca}
}
\date{}

\def\eps{\varepsilon}
\def\F{{\mathbb F}}
\def\Q{{\mathbb Q}}
\def\C{{\mathbb C}}
\def\Z{{\mathbb Z}}
\def\Fn{{\mathbb F_{2^n}}}
\def\e{\eta}

\newcommand{\Ft}{\mathbb{F}_{2}}
\newcommand{\W}{{\mathcal W}}
\newcommand{\LL}{{\mathcal L}}

\def\P{\phantom{\Big|}}

\def\r#1 {{\color{red} #1} }
\def\wt{{\rm wt}}

\newtheorem{theorem}{Theorem}[section]
\newtheorem{lemma}[theorem]{Lemma}
\newtheorem{proposition}[theorem]{Proposition}
\newtheorem{corollary}[theorem]{Corollary}
\newtheorem{conjecture}[theorem]{Conjecture}
\theoremstyle{definition}
\newtheorem{definition}[theorem]{Definition}
\newtheorem{example}[theorem]{Example}
\newtheorem{remark}[theorem]{Remark}
\newtheorem{problem}[theorem]{Problem}

\def\ra{\rightarrow}

\def\gf{\F}
\def\tr{{\rm Tr}}
\def\Tr{{\rm Tr}}

\def\t{\theta}

\def\w{\omega}

\maketitle

\begin{abstract}
	A Walsh zero space (WZ space) for $f:\Fn\rightarrow\Fn$
    is an $n$-dimensional vector subspace of $\Fn\times\Fn$
    whose all nonzero elements are Walsh zeros of $f$.
    We provide several theoretical and computer-free
    constructions of WZ spaces for Gold APN functions $f(x)=x^{2^i+1}$ on $\Fn$
    where $n$ is odd and $\gcd(i,n)=1$.
    We also provide several
    constructions of trivially intersecting
    pairs of such spaces.
    We illustrate applications
    of our constructions that include constructing
    APN permutations that are CCZ equivalent to $f$
    but not extended affine equivalent to $f$ or
    its compositional inverse.
\end{abstract}

\section{Introduction}

    One of the most important open problems concerning APN permutations is their existence in even dimensions. Dillon et al.~\cite{BDMW} constructed an APN permutation of $\F_{2^6}$ in 2009 but since then, no new APN permutations in even dimension have been found. The question of existence of APN permutations in even dimension greater than six is known as the Big APN Problem.

    For odd dimensions, there are six known infinite families 
    of monomial APN permutations, see Table 1 of \cite{BCC}. 
    An important open problem is the conjecture by Dobbertin \cite{HD}
    that any monomial APN function is CCZ-equivalent to a 
    permutation belonging to one of these known families.  

    Currently, all known APN permutations in odd dimensions, 
    up to CCZ equivalence, belong to a monomial family or one of the following families:
    (1) An infinite family of quadratic APN permutations in odd 
    dimension found in 2008 by Beierle, Carlet, and Leander~\cite{BCL}.
    (2) Two quadratic APN permutations in dimension~$9$ found in 2020 by Beierle and Leander~\cite{BL}.

    As such,
    it is certainly still desirable to find new constructions
    of APN permutations in odd dimensions as well.
    This is one of the main goals of this paper.
    The other goal are theoretical descriptions
    of Walsh zero (WZ) spaces of functions.
    Some WZ spaces of Gold functions in odd dimensions
    are constructed
    computer-free, inspired by concrete examples
    in low dimensions. 
    WZ spaces already occur implicitly in the work
    of Dillon et al.~\cite{BDMW}; recently their
    importance and applications were recognized more
    explicitly, see \cite{CP} and references therein.

\section{Preliminaries}
    Let $\Fn$ denote the finite field with $2^n$~elements
    and let $\F_{2^n}^*$ denote the set of its nonzero elements.
    A function $f:\F_{2^n}\rightarrow\F_{2^n}$
    is {\em almost perfect nonlinear (APN)} if for all
    $a,b\in\F_{2^n}$, $a\neq 0$, the equation
    $f(x+a)-f(x)=b$ has at most two solutions $x\in\F_{2^n}$.
    Without loss of generality, we can normalize
    any APN function such that $f(0)=0$, and we will assume
    this throughout.
    
    APN functions,
    and more generally functions
    with low differential uniformity, 
    have been extensively studied
    due to their importance in the design of 
    S-boxes of block ciphers in cryptography,
    where they offer the best possible protection
    against differential cryptanalysis.
    In some block cipher designs, such as substitution-permutation networks (SPN)
    \cite[Chapter~4]{StPa},
    it is required that S-boxes are invertible mappings.
    Of special interest are 
    therefore APN functions which are invertible, that is,
    they are {\em permutations} of $\F_{2^n}$.
    Constructing new APN permutations
    of $\F_{2^n}$ is one of the objectives of
    our work.  
    
    Let $\Tr^n_m$ denote the trace function from $\Fn$ to $\F_{2^m}$,
    and let $\Tr$ denote the absolute trace function from $\Fn$ to $\F_2$.
    Let $f$ be a function from $\Fn$ to $\Fn$.
    For $(a,b)\in\Fn\times\Fn$ 
    we define the Walsh transform of $f$ at $(a,b)$
    as $\W_f(a,b)=\sum_{x\in\Fn} (-1)^{\Tr(ax+bf(x))}$.
    We say that $(a,b)$ is a {\em Walsh zero} of $f$
    if $\W_f(a,b)=0$. 
    
    \begin{definition}
        Let $f$ be a function from $\Fn$ to $\Fn$.
        Suppose that $S$ is an $\Ft$-linear subspace of $\Fn\times\Fn$
        such that $\dim_{\Ft} S=n$ and each element of $S$ except $(0,0)$
        is a Walsh zero of $f$. We say that $S$ is a {\em WZ space} of~$f$.
    \end{definition}
    
    We say that two WZ spaces $S,T$ of the same function
    {\em intersect trivially} if $S\cap T=\{(0,0)\}$.

    We will need some background from coding theory. 
    The Hamming weight $\wt(x)$ of a vector $x\in \F_2^n$ is the number of 
    nonzero coordinates in $x$.
    A binary linear code $C$ with parameters $[n,k,d]$ is an $\F_2$-linear subspace of $\F_2^n$ where $\dim_{\F_2}C = k$ and $d$ is the minimum Hamming weight of nonzero codewords in $C$. We say that $d$ is the distance of $C$.
    A generator matrix for $C$ is a $k\times n$ matrix whose rows form a basis for $C$.  
    Permuting the coordinates of a  code $C$ will produce a code $C'$ with identical parameters. We say that such codes $C$ and $C'$ are equivalent. 

        For a positive integer $r$, the binary simplex code $\mathcal{S}_r$ is a binary linear code with parameters $[2^r-1,r,2^{r-1}]$.
    It is well known that
    this definition is sound, that is, the parameters
    determine the simplex code uniquely up to equivalence.
Moreover
     for every nonzero $x\in S_r$ the Hamming weight of $x$ is $2^{r-1}$.
    A generator matrix for the code $S_r$ is an $r\times (2^r - 1)$ matrix whose columns consist of all distinct nonzero vectors of $\F_2^r$. 
    
    The {\em CCZ equivalence} of functions was introduced
    by Carlet, Charpin and Zinoviev in \cite{CCZ}.
    It has many important features, in particular it preserves
    the APN property. 
    Dillon et al.\ introduced in \cite{BDMW} a method that,
    assuming certain conditions are satisfied, constructs
    a permutation that is CCZ equivalent to a given function.
    In the following proposition
    we present this method in a different but equivalent form,
    using the concept of WZ spaces. We also include a proof
    of the proposition, which is contained only implicitly in \cite{BDMW},
    because it allows one to {\em explicitly construct} a permutation CCZ
    equivalent to the given function. 
    
    \begin{proposition}\label{prop-wz-disj} 
        Let $f$ be a function from $\Fn$ to $\Fn$. If there exist two WZ spaces of $f$ that intersect trivially, then $f$ is CCZ-equivalent to a permutation of $\Fn$.
    \end{proposition}
    \begin{proof}
    	Without loss of generality we can assume
    	that $f(0)=0$, since this can be achieved by
    	an affine transformation which is a special kind
    	of CCZ equivalence.
    	
    	Let $\{\alpha_1,\ldots,\alpha_n\}$ and $\{\beta_1,\ldots,\beta_n\}$ be two dual bases
    	of $\F_{2^n}$ over $\F_2$, that is,
    	$\Tr(\alpha_i\beta_j)=\delta_{i,j}$
    	where for $1\le i,j\le n$ we let $\delta_{i,j}=1$
    	if $i=j$ and $\delta_{i,j}=0$ otherwise.
        Let the nonzero elements of $\F_{2^n}$
        be labelled $x_1,\ldots,x_{2^n-1}$. 
         
        Let $G=\left(\begin{array}{c} G_1\\ G_2\end{array}\right)$ 
        be a $(2n)\times(2^n-1)$ matrix over $\F_2$
        where $G_1$ and $G_2$ are
        $n\times(2^n-1)$ matrices over $\F_2$
        defined as follows. The entry in row $i$
        and column $j$ of $G_1$ is $\Tr(\alpha_ix_j)$. 
         The entry in row $i$
        and column $j$ of $G_2$ is $\Tr(\alpha_i f(x_j) )$. 
        Then the $j$-th column of $G$ is of the form
        $\left(
        \begin{array}{c}
        x_j \\ f(x_j)
        \end{array}\right)
        $ where $x_j$ and $f(x_j)$ are represented
        as $n$-dimensional column vectors with respect
        to the basis $\{\beta_1,\ldots,\beta_n\}$
        introduced above. Let $C$ be the binary linear
        code which is the row space of $G$. Each
        codeword of $C$ is of the form
        $\Tr(rx_j+sf(x_j))_{j=1,\ldots, 2^n-1}$
        where $r,s$ are fixed elements of $\F_{2^n}$.

        Let $S$ and $T$
        be the two given trivially intersecting WZ spaces,
        and let
        $B_1=\{(a_1,b_1),\ldots,(a_n,b_n)\}$ and $B_2=\{(a_{n+1},b_{n+1}),\ldots,(a_{2n},b_{2n})\}$
        be their bases. 
        Note that  $B_1\cup B_2$
        is a basis for $\F_{2^n}\times\F_{2^n}$. 
        
        Let $G'=\left(\begin{array}{c} G'_1\\ G'_2\end{array}\right)$ 
        be a $(2n)\times(2^n-1)$ matrix over $\F_2$
        where $G'_1$ and $G'_2$ are
        $n\times(2^n-1)$ matrices over $\F_2$
        defined as follows. The entry in row $i$
        and column $j$ of $G'_1$ is $\Tr(a_ix_j+b_i f(x_j))$. 
        The entry in row $i$
        and column $j$ of $G'_2$ is 
        $\Tr(a_{n+i}x_j+b_{n+i} f(x_j))$. 
        For $i=1,2$ let $C_i'$ be the row space of $G_i'$.   
        Since $S$ and $T$ are WZ spaces for $f$,
        each nonzero codeword of $C_i'$ has weight $2^{n-1}$
        for $i=1,2$. Thus $C_1'$, $C_2'$ are simplex codes $S_n$,
        and
        $G_1'$, $G_2'$ 
        are formed by pairwise distinct nonzero columns.
        Thus, with respect to the basis $\{\beta_1,\ldots,\beta_n\}$,
         the columns of $G'$ can be viewed as
        $\left(
        \begin{array}{c}
        x \\ g(x)
        \end{array}\right)$
        where $x$ runs through all nonzero elements of $\F_{2^n}$
        and $g(x)\in\F_{2^n}$.
        After letting $g(0)=0$ we
        see that $g$ is a permutation of $\F_{2^n}$.
        
        Let $C'$ be the rowspace of $G'$. 
        Since $B_1\cup B_2$
        is a basis for $\F_{2^n}\times\F_{2^n}$,
        each
        codeword of $C'$ is of the form
        $\Tr(rx_j+sf(x_j))_{j=1,\ldots, 2^n-1}$
        where $r,s$ are fixed elements of $\F_{2^n}$.
        Thus the codes $C$ and $C'$ are equal,
        and functions $f$ and $g$ are CCZ equivalent
        by \cite{Dlong}, see also paragraph~4 on page 380
        in~\cite{Car21}.

    \end{proof}

\section{Walsh zeros of Gold functions}

    In this section we characterize Walsh zeros of Gold APN permutations.
    That is, we work in odd dimension throughout. We also define and give some results on certain additive subspaces of $\F_{2^n}$ which we call \emph{compatible subspaces}. These spaces appear
    while studying Walsh zero spaces of Gold APN permutations. 

    The following proposition  was first proved implicitly by 
     Gold in \cite{Gold1968}.

    \begin{proposition}\label{WZtest}
    Suppose $n$ is odd and $f:\F_{2^n}\rightarrow \F_{2^n}$ with $f(x) = x^{2^i + 1}$ such that $\gcd(i,n) = 1$ (so $f$ is a Gold APN function). Then $(a,b)\in\Fn\times\Fn$ is a Walsh zero of $f$ if and only if $\Tr(ab^{-\frac{1}{2^i + 1}}) = 0$ or $a\ne b = 0$.
    \end{proposition}
    \begin{proof}
        First we present a proof given in \cite{LMW} which covers the case $b=1$.
        We then extend it to $b\neq 1$. The notation $\chi(y)=(-1)^{\Tr(y)}$ is introduced only to improve readability. Let $f:\F_{2^n}\rightarrow\F_{2^n}$ be given by $f(y) = y^{2^i+1}$ with $\gcd(i,n)=1$. Then for any $c\in\F_{2^n}$,
        \begin{align*}
            \mathcal{W}_f(A,1) &= \sum_{y\in\F_{2^n}}(-1)^{\Tr(Ay + f(y))}\\ 
            &= \sum_{z\in\F_{2^n}}\chi((z + c )^{2^i+1} + A(z + c))\\
            &= \sum_{z\in\F_{2^n}}\chi(z^{2^i+1} + cz^{2^i} + c^{2^i}z + c^{2^i+1} + Az + Ac)\\
            &= \chi(Ac + c^{2^i+1})\sum_{z\in\F_{2^n}}\chi(z^{2^i+1} + c^{2^{-i}}z + c^{2^i}z + Az)\\
            &= \chi(Ac + c^{2^i+1})\sum_{z\in\F_{2^n}}\chi(z^{2^i+1} + z(L(c) + A))
        \end{align*}
        where $L(c) = c^{2^i} + c^{2^{-i}}.$

        It is clear that $L$ is linear and that $\Tr(L(c))=0$ for all $c$. Since the kernel of $L$ is $\{ 0,1 \}$, the image of $L$ contains all elements of $\F_{2^n}$ with trace $0$. 
        So if $\Tr(A) = 0$ we can chose $c$ such that $A=L(c)$.
        With this choice of $c$, we have 
        $$W_f(A,1) = \chi(Ac + c^{2^i+1})\sum_{z\in\F_{2^n}}\chi(z^{2^i+1}) = 0$$
        where $\sum_{z\in\F_{2^n}}\chi(z^{2^i+1}) = 0$ since $z^{2^i+1}$ is a permutation of $\F_{2^n}$ when $n$ is odd.

        If $\Tr(A) = 1$, choose $c$ so that $L(c) = A + 1$. Then
        $$W_f(A,1) = \chi(Ac + c^{2^i+1})\sum_{z\in\F_{2^n}}\chi(z^{2^i+1} + z) = \chi(Ac + c^{2^i+1})W_f(1,1).$$ The following is proved as Theorem 5 of \cite{LMW},
        $$W_f(1,1) = \begin{cases}
            +2^{(n+1)/2} & \text{if }n\equiv \pm 1\text{ (mod 8)}\\ 
            -2^{(n+1)/2} & \text{if }n\equiv \pm 3\text{ (mod 8)} .
        \end{cases}$$
        So $W_f(A,1)$ is nonzero when $\Tr(A) = 1$. Therefore $\mathcal{W}_f(A,1)=0$ if and only if $\Tr(A)=0$.

        This result is extended to $b\neq 0, 1$ with a change of variables. Let $z=b^\frac{1}{2^i+1}x$. Then \begin{align*}
            \mathcal{W}_f(A,1) &= \sum_{z\in\F_{2^n}}\chi(Az + f(z))\\
            &= \sum_{x\in\F_{2^n}}\chi(Ab^\frac{1}{2^i+1}x + bf(x)) \\
            &= \mathcal{W}_f(a,b) 
        \end{align*}
        where $A$ is chosen such that $a=Ab^\frac{1}{2^i+1}$. The condition $\Tr(A)=0$ becomes $\Tr(ab^{-\frac{1}{2^i+1}})=0$. Finally, if $a\ne b = 0$ then $\mathcal{W}_f(a,b)=\sum_{x\in\F_{2^n}}\chi(ax)=0.$

    \end{proof}

    If it is clear from the context that $f$ is a Gold APN permutation, 
    then we refer to Proposition~\ref{WZtest} simply 
    as \emph{the Walsh zero test}.

    Suppose $f:\F_{2^n} \rightarrow \F_{2^n}$.
    Since $\Tr(ax)$ is a balanced function of $x$, 
    the space $$Z_{a0} = \{(a,0) : a\in\F_{2^n} \}$$ 
    is a WZ space of $f$. Furthermore, if $f$ is a permutation and $b\neq 0$
    then $\Tr(bf(x))$ is a balanced function 
    of $x$, hence and $$Z_{0b} =  \{(0,b) : b\in\F_{2^n} \}$$ is 
    also a WZ space of $f$. 
    We call these two spaces \emph{trivial} WZ spaces.

    WZ spaces have been recently appearing in literature.
    In Theorem 4 of \cite{CP}, Canteaut and Perrin prove that 
    the number of EA-classes inside the CCZ-class of $f:\Fn\rightarrow\F_{2^m}$
    is upper bounded by the number of WZ spaces of $f$. 
    Beierle, Carlet, Leander, and Perrin in \cite{BCLP}
    have studied two new quadratic APN permutations in dimension 9 
    using numerical properties of their WZ spaces (thickness and degree spectrum). 
    Furthermore, in \cite{BCLP} the authors show an interesting 
    similarity between their new APN permutations and Gold APN 
    permutations in odd dimension divisible by 3. 
    In the next section we provide constructions of nontrivial 
    WZ spaces for Gold APN permutations in odd dimension which are, 
    as far as we know, the first such theoretical 
    (computer-free)
    constructions.
    
    From now on let us define $0^{-\frac{1}{2^i + 1}}=0$ as this will simplify
    some of the forthcoming constructions and arguments. Let us note that
    this works correctly for the Walsh zero test for Gold APN permutations,
    where for a pair $(a,0)$ we get $\Tr(a\cdot 0^{-\frac{1}{2^i + 1}})
    =\Tr(a\cdot 0)=0$, as desired.
    
    \begin{definition}
        Let $n$ be odd and $\gcd(i,n)=1$. Let $S$ be an additive subspace of $\Fn$. We say that $S$ is {\em $i$-compatible} if the set $S^{-\frac{1}{2^i + 1}}=\{ s^{-\frac{1}{2^i + 1}} \::\: s\in S \}$ is also an additive subspace of $\Fn$.
    \end{definition}




        

    For $U\subseteq \Fn$ and $a\in\Fn$
    denote $aU=\{au \: : \: u \in U\}$.
    
    \begin{example}
        \label{ex-triv-compa}
        \ \\
        (i)
        If $n$ is odd and $\gcd(i,n)=1$ then the following subspaces
        of $\Fn$ are $i$-compatible: $\{0\}$, $\F_2$
        and $\Fn$.\\
        (ii)
        If $S$ is an $i$-compatible subspace of $\Fn$, 
        then $\mu S$ is also $i$-compatible
        for each $\mu\in \Fn$.
    \end{example}

    \begin{proposition}
    	\label{prop-subfield-compatible}
        Suppose $m$ divides $n$ and $\gcd(i,n) = 1$. Then $\F_{2^m}$ is $i$-compatible subspace of $\F_{2^n}$.
    \end{proposition}
    \begin{proof}
        Let $\xi\in\F_{2^m}$. Then 
        $(\xi^{-\frac{1}{2^i+1}})^{2^m} = (\xi^{2^m})^{-\frac{1}{2^i+1}} 
        = \xi^{-\frac{1}{2^i+1}}$.
        So $\xi^{-\frac{1}{2^i+1}} \in\F_{2^m}$ and it follows that 
        $\{ \xi^{-\frac{1}{2^i+1}} \; : \; \xi\in\F_{2^m} \}
        = \F_{2^m}$.

    \end{proof}
    
    Moreover, when $n$ is a multiple of~$3$,
    then $\Fn$ contains the subfield $\F_{2^3}$
    and the following lemma applies.

    \begin{lemma} \label{f8linear}
        Suppose $n$ is a multiple of $3$ and $\gcd(i,n)=1$. Let $\xi\in\F_{2^3}\subset \F_{2^n}$. Then $\xi^{-\frac{1}{2^i+1}} = \xi^{2^i}$. 
    \end{lemma}
    \begin{proof}
        For $\xi = 0$ the result follows from our definition $0^{-\frac{1}{2^i+1}} = 0$. 
        Otherwise, after raising both sides of $\xi^{-\frac{1}{2^i+1}} = \xi^{2^i}$ to power
        $2^i+1$ we face proving the equivalent equation $$\xi^{-1} = \xi^{2^{2i} + 2^i}.$$
For $i\equiv 1\pmod{3}$
this follows from 
\[
\xi^{2^{2i} + 2^i} = \xi^{2^{2} + 2^1}
            = \xi^6 = \xi^{-1}
            \]
and for $i\equiv 2\pmod{3}$
it follows from 
\[
\xi^{2^{2i} + 2^i} = \xi^{2^1 + 2^2}
= \xi^6 = \xi^{-1}.
\]
    \end{proof}
    
    \begin{corollary}
        \label{f8}
        Assume that $n$ is odd and divisible by 3. Let $S$ be the subfield of $\F_{2^n}$ isomorphic to $\F_{2^3}$.
        Then each additive subspace of $S$ is $i$-compatible whenever $\gcd(i,n)=1$.
    \end{corollary}
    \begin{proof}
        This follows from Lemma \ref{f8linear} since $\xi\mapsto \xi^{-\frac{1}{2^i + 1}}$
        is a linear mapping on $\F_{2^3}$.
    \end{proof}

    As any two 2-dimensional subspaces (hyperplanes) of $\F_{2^3}$
    can be obtained from each other just by scaling by
    an element of $\F_{2^3}^*$, it follows that if $3$ divides~$n$,
    then Lemma~\ref{f8} provides 
    $2^n-1$ two-dimensional
    $i$-compatible subspaces of $\Fn$
    and
    $(2^n-1)/7$ three-dimensional
    $i$-compatible subspaces of $\Fn$.
    We have checked by exhaustive search that there are no two-dimensional 
    $i$-compatible subspaces of $\Fn$ other than those described above
    for odd $n$ and relatively prime $i$
    in the range $n\le 17$.

    For $i=1$ it is easy to see that the proof of Lemma~\ref{f8} does not
    extend to any higher dimension.
    Indeed suppose that $s\mapsto s^{-1/3}$ is linear on $\Fn$
    where $n>3$ is odd, then $-1/3 \equiv 2^k\pmod{2^n-1}$
    for some integer $0\le k<n$, and 
    $3\cdot 2^k+1$ is a multiple of $2^n-1$.
    If $n>3$ and $k< n-1$ then $3\cdot 2^k+1$ is too small
    for this to happen. The only remaining
    possibility is $3\cdot 2^{n-1}+1=2^n-1$, which can not occur either.
    Similarly, to extend Lemma~\ref{f8} to higher dimension for $i>1$,
    we must show that $2^k(2^i + 1) + 1$ is a multiple of $2^n-1$ for 
    some integer $0\le k<n$. But a parity argument shows that this is also 
    not possible.

In Proposition~\ref{prop-comp} below we will see
that $i$-compatible spaces can be used to construct
WZ spaces. Thus it is interesting to obtain
as many $i$-compatible spaces as we can. This motivates
us to present the following open problem:

    \begin{problem}
        For odd $n$, do there exist $i$-compatible subspaces
        of $\Fn$ other than those described by Example~\ref{ex-triv-compa}(i,ii),
        Proposition~\ref{prop-subfield-compatible}
        and Corollary~\ref{f8}?
    \end{problem}

\section{Constructions of some WZ spaces for Gold APN permutations}
\label{section-WZspace}
    
    Theoretical constructions in 
    this section have been
    partially motivated by computational examples
    in low odd dimensions obtained using the sboxU
    software package \cite{sboxU,BPT}
    written by
    L\'eo Perrin and Mathias Joly.
    In particular, 
    we verified using sboxU
    that the theoretical constructions
    given in this section, along with
    the two trivial spaces $Z_{a0}$ and $Z_{0b}$, cover
    all WZ spaces of Gold APN functions
    in all odd dimensions less than or equal to~$9$.

    \begin{proposition}
        \label{prop-comp}
        Assume that $n$ is odd and $S$ is an $i$-compatible additive
        subspace of $\Fn$.
        Let $f:\Fn\rightarrow \Fn$, $f(x)=x^{2^i + 1}$, with $\gcd(i,n)=1$.
        Let $$X=\{ x \in \Fn \: : \:
        (\forall a\in S^{-\frac{1}{2^i + 1}}) \; \Tr(ax)=0   \}.$$
        Then $X\times S$ is a WZ space for~$f$.
    \end{proposition}
    \begin{proof}
        Let $(0,0)\neq (x,s)\in X\times S$. 
        Then $\Tr(xs^{-\frac{1}{2^i + 1}})=0$ by construction, hence    
        $(x,s)$ is a Walsh zero of~$f$. Since both $X$ and $S$
        are linear spaces, $X\times S$ is also a linear space.
        Finally, since $s\mapsto s^{-\frac{1}{2^i + 1}}$ is a bijection on $\Fn$,
        we get $\dim S^{-\frac{1}{2^i + 1}}=\dim S $, and
        \[
        \dim( X\times S) = \dim X + \dim S = (n-\dim S)+\dim S =n. 
        \]
    \end{proof}
    
    \begin{proposition} 
        \label{prop-xi}
        Assume that $n=3k$ where $k$ is odd.
        Let $f:\Fn\rightarrow \Fn$, $f(x)=x^{2^i+1}$, with $\gcd(i,n)=1$.
        Suppose $\xi$ is a fixed element of $\F_{2^3}\subset\F_{2^n}$ and $\mu$ is a fixed element of $\F_{2^n}^*$. Then 
        \[
        Z = \left\{ \left(x \, , \, \mu^{-(2^i+1)}( \xi\Tr(\mu x) + \Tr(\xi^{2^i}\mu x)) \right) \ : \ x \in \F_{2^n}
        \right\}
        \]
        is a WZ space of $f$.
    \end{proposition}
    \begin{proof}
        The additive closure of $Z$ follows from the fact
        that $$g(x) = \mu^{-(2^i+1)}( \xi\Tr(\mu x) + \Tr(\xi^{2^i}\mu x))$$ is an
        $\F_2$-linear function.
        
        We now prove that each element $(x,g(x)) \in Z \backslash \{ (0,0) \}$
        is a Walsh zero of $f$.
        We only need to address the cases when $ g(x)\neq 0 $.
        Note that $ \mu^{2^i+1} g(x)\in \F_{2^3} $. 
        Using Lemma~\ref{f8linear} we get
        \begin{eqnarray*}
            \Tr(xg(x)^{-\frac{1}{2^i + 1}} )
            & = &
            \Tr(x\mu(\mu^{2^i+1}g(x) )^{-\frac{1}{2^i + 1}}) \\
            & = &
            \Tr(x\mu(\xi\Tr(\mu x) + \Tr(\xi^{2^i}\mu x))^{2^i}) \\
            & = &
            \Tr(x\mu\xi^{2^i}\Tr(\mu x) + x\mu\Tr(\xi^{2^i}\mu x)) \\
            & = &
            \Tr(x\mu\xi^{2^i}\Tr(\mu x)) + \Tr(x\mu\Tr(\xi^{2^i}\mu x)) \\
            & = & 0
        \end{eqnarray*}
        because $\Tr(x \mu\xi^{2^i}\Tr(\mu x)) = \Tr(x \mu\xi^{2^i})\Tr(\mu x) = \Tr(x\mu \Tr(\xi^{2^i}\mu x))$. By considering the first components
        of the elements of $Z$ it is clear that $\dim Z =n$.
    \end{proof}
    While Proposition \ref{prop-xi} holds for each $\xi\in \F_{2^3}$, one obtains interesting results only when $\xi$ is a primitive element of $\F_{2^3}$. If $\xi = 0,1$ then $Z$ is the trivial WZ space $Z_{a0}$.
    
    \begin{proposition} 
        \label{prop-traceFm}
        Suppose $n$ is odd and $m$ divides $n$. 
        Let $f:\F_{2^n} \rightarrow \F_{2^n}$, $f(x) = x^{2^i+1}$, with $\gcd(i,n)=1$. Suppose $\mu$ is a fixed element of $\F_{2^n}^*$. Then 
        \[
        Z = \{ (\mu a, \mu^{2^i+1} b) : a\in \F_{2^n}, b\in \F_{2^m} \textup{ and } \Tr_m^n(a) = b + b^{\frac{1}{2^i}} \}
        \]
        is a WZ space of $f$.
    \end{proposition}
    \begin{proof}
        First we show that each element $(\mu a, \mu^{2^i+1} b)$ of $Z$ except for $(0,0)$ is a Walsh zero of $f$. Note that since $m$ is odd, $b^{-\frac{1}{2^i+1}} \in \F_{2^m}$. We have
        \begin{eqnarray*}
            \Tr_1^n(\mu a(\mu^{2^i+1} b)^{-\frac{1}{2^i+1}})
            & = &
            \Tr_1^n(ab^{-\frac{1}{2^i+1}}) \\
            & = &
            \Tr_1^m(\Tr_m^n(ab^{-\frac{1}{2^i+1}})) \\
            & = &
            \Tr_1^m(b^{-\frac{1}{2^i+1}}\Tr_m^n(a)) \\
            & = &
            \Tr_1^m(b^{-\frac{1}{2^i+1}}(b + b^{\frac{1}{2^i}})) \\
            & = &
            \Tr_1^m(b^{\frac{2^i}{2^i+1}} + b^{\frac{1}{2^i(2^i+1)}}) \\
            & = &
            \Tr_1^m(b^{\frac{2^i}{2^i+1}}) + \Tr_1^m\left(b^{\frac{2^{2i}}{2^i(2^i+1)}}\right) \\
            & = &
            0
        \end{eqnarray*}
        hence each nonzero element of $Z$ is a Walsh zero
        of $f$. 
        
        Let $(\mu a_1 , \mu^{2^i+1} b_1),(\mu a_2, \mu^{2^i+1} b_2)\in Z$. Then 
        \[
        (\mu a_1 , \mu^{2^i+1} b_1) + (\mu a_2, \mu^{2^i+1} b_2) 
        = (\mu (a_1+a_2), \mu^{2^i+1}(b_1 + b_2))\in Z
        \]
        since $b_1 + b_2 \in \F_{2^m}$ and \begin{eqnarray*}
            \Tr_m^n(a_1 + a_2) 
            & = &
            \Tr_m^n(a_1) + \Tr_m^n(a_2) \\
            & = &
            b_1 + b_1^{\frac{1}{2^i}} + b_2 + b_2^{\frac{1}{2^i}} \\
            & = &
            (b_1 + b_2) + (b_1 + b_2)^{\frac{1}{2^i}}.
        \end{eqnarray*}
        Therefore $Z$ is additively closed.
        
        For each of the $2^m$ choices for the second component $\mu^{2^i+1} b$ of an element of $Z$, there are $2^{n-m}$ choices for the first component $\mu a$ such that $\Tr_m^n(a) = b + b^{\frac{1}{2^i}}$. It follows that the dimension of $Z$ is $n$.  
        
    \end{proof}

    \section{Trivially intersecting pairs of WZ spaces for Gold APN permutations}
    \label{sec-Ti-pairs}
    In this section we describe 
    several theoretical constructions
    of trivially intersecting
    pairs of WZ spaces for Gold APN permutations in odd dimensions.
    We checked using the sboxU software package
    that these constructions cover
    all such pairs in dimensions less than or equal to~9.
    Obviously, the two trivial WZ spaces $Z_{a0},Z_{0b}$
    intersect trivially.

    \begin{proposition} \label{prop-TI-0,n_n,2}
        Let $n=3k$ where $k$ is odd and $f:\F_{2^n} \rightarrow \F_{2^n}$, $f(x) = x^{2^i+1}$, with $\gcd(i,n)=1$. Let $$Z=\{ (x,\mu^{-(2^i+1)}(\xi\Tr(\mu x)+\Tr(\xi^{2^i}\mu x))) : x \in \F_{2^n}\}$$ where $\xi$ is a fixed element of $\F_{2^3}\subset\F_{2^n}$ and $\mu$ is a fixed element of $\F_{2^n}^*$. Then $Z$ is a WZ space of $f$ and the pair $\{ Z_{0b}, Z \}$ intersects trivially. 
    \end{proposition}
    \begin{proof}
        It follows from Proposition \ref{prop-xi} that $Z$ is a WZ space of $f$. If $(x,c)\in Z_{0b}\cap Z$,
        then $x=0$ which implies $c=0$.
        Therefore the pair $\{ Z_{0b}, Z \}$ intersects trivially.
    \end{proof}

    Note that in the above proposition, if $\xi\in\{ 0,1 \}$ then $Z$ 
    becomes the trivial WZ space $Z_{a0}$.
    
    \begin{proposition} \label{TI-hard-1}
        Let $n=3k$ where $k$ is odd and $f:\F_{2^n} \rightarrow \F_{2^n}$, $f(x) = x^{2^i+1}$, with $\gcd(i,n)=1$.
        Let $$Y=\{ (x,\mu^{-(2^i+1)}(\xi\Tr(\mu x)+\Tr(\xi^{2^i}\mu x))) : x \in \F_{2^n}\}$$  where $\xi$ is a fixed element of $\F_{2^3}\subset\F_{2^n}$ 
        and $\mu$ is a fixed element of $\F_{2^n}^*$.
        Let $$Z = \{(\nu (b + b^{\frac{1}{2^i}}), \nu^{2^i+1} b) : b\in \F_{2^n} \}$$ with $\nu$ a fixed element of $\F_{2^n}^*$.
        Suppose also that $\Tr((\xi + \xi^{2^i})(\mu\nu)^{-2^i}) = 0$.  
        Then $Y$ and $Z$ are WZ spaces of $f$ and the pair $\{ Y,Z \}$ intersects trivially.
    \end{proposition}
    \begin{proof}
        It follows from Proposition \ref{prop-xi} that $Y$ is a WZ space of $f$ 
        and from Proposition \ref{prop-traceFm} with $m=n$ that $Z$ is a WZ space of $f$.

        Suppose towards a contradiction that there exists an element, different from $(0,0)$, in $Y\cap Z$. Then for this element, we have $x = \nu(b + b^{\frac{1}{2^i}})$. It follows that $b\neq 0$ and 
        \begin{equation}
            \label{b-comp}
            (\mu\nu)^{2^i+1} b = \xi\Tr(\mu\nu(b+b^{\frac{1}{2^i}})) + \Tr(\xi^{2^i}\mu\nu(b+b^{\frac{1}{2^i}})).
        \end{equation}
        By looking at the right-hand side of \eqref{b-comp} we can see that $(\mu\nu)^{2^i+1}b \in \{ 0,1,\xi, \xi + 1 \}$. But since $\mu \neq 0$, $\nu \ne 0$, and $b \ne 0$, we have $(\mu\nu)^{2^i+1}b \in \{ 1,\xi, \xi+1 \}$. 
        We will look at these three cases below.
        
        First assume that
        $(\mu\nu)^{2^i+1}b = 1$. Observe that
        \begin{align*}
            \Tr(\xi^{2^i}\mu\nu(b + b^{\frac{1}{2^i}})) 
            &= \Tr(\xi^{2^i}\mu\nu(\mu\nu)^{-(2^i+1)} + \xi^{2^i}\mu\nu(\mu\nu)^{-\frac{2^i+1}{2^i}})\\
            &= \Tr(\xi^{2^i}(\mu\nu)^{-2^i}) + \Tr(\xi^{2^i}(\mu\nu)^{-\frac{1}{2^i}})\\
            &= \Tr(\xi^{2^i}(\mu\nu)^{-2^i}) + \Tr\!\left(\xi^{2^{3i}}(\mu\nu)^{-\frac{2^{2i}}{2^i}}\right)\\
            &= \Tr(\xi^{2^i}(\mu\nu)^{-2^i}) + \Tr(\xi^{2^{3i}}(\mu\nu)^{-2^i})\\
            &= \Tr((\xi^{2^i} + \xi)(\mu\nu)^{-2^i}) = 0
        \end{align*} 
        which contradicts \eqref{b-comp}.

        Finally assume that
        $(\mu\nu)^{2^i+1}b =z$
        and $z\in\{\xi,\xi+1\}$. Then
        \begin{align*}
            \Tr(\mu\nu(b + b^{\frac{1}{2^i}})) 
            &= \Tr(z\mu\nu(\mu\nu)^{-(2^i+1)} + z^{\frac{1}{2^i}}\mu\nu(\mu\nu)^{-\frac{2^i+1}{2^i}})\\
            &= \Tr(z(\mu\nu)^{-2^i}) + \Tr(z^{\frac{1}{2^i}}(\mu\nu)^{-\frac{1}{2^i}})\\
            &= \Tr(z(\mu\nu)^{-2^i}) + \Tr\!\left(z^{\frac{2^{2i}}{2^i}}(\mu\nu)^{-\frac{2^{2i}}{2^i}}\right)\\
            &= \Tr(z(\mu\nu)^{-2^i}) + \Tr(z^{2^i}(\mu\nu)^{-2^i})\\
            &= \Tr((\xi + \xi^{2^i})(\mu\nu)^{-2^i}) = 0
        \end{align*} 
        which again contradicts \eqref{b-comp}.

        Since all cases are exhausted, the only element in the intersection is $(0,0)$ and 
        the proof is complete. 
    \end{proof}

\begin{proposition} \label{TI-hard-2}
        Let $n=3k$ where $k$ is odd and $f:\F_{2^n} \rightarrow \F_{2^n}$, $f(x) = x^{2^i+1}$, with $\gcd(i,n)=1$.
        Let $$Y = \{(\nu (b + b^{\frac{1}{2^i}}), \nu^{2^i+1} b) : b\in \F_{2^n} \}$$ with $\nu$ a fixed element of $\F_{2^n}^*$. 
        Suppose $\xi$ is a fixed primitive element of $\F_{2^3}\subset\F_{2^n}$. 
        Let $$Z = X \times S$$ with $S = \textup{span}_{\F_2}\{ \mu, \xi\mu \}$ for some fixed $\mu\in\F_{2^n}^*$ and $Z$ is constructed by applying Proposition \ref{prop-comp}.
        Suppose also that $\Tr( (\xi + \xi^{2^i})\mu^\frac{2^i}{2^i+1}\nu^{-2^i} ) = 1.$
        Then $Y$ and $Z$ are WZ spaces of $f$ and the pair $\{ Y,Z \}$ intersects trivially.
    \end{proposition}
    \begin{proof}
        It follows from Proposition \ref{prop-traceFm} with $m=n$ that $Y$ is a WZ space of $f$.

        Note that $S$ is $i$-compatible by Lemma~\ref{f8}
        and Example~\ref{ex-triv-compa}(ii):
         If $T = \textup{span}_{\F_2}\{1,\xi \}$ then $S = \mu T$. From Proposition \ref{prop-comp} we have 
        $$X = \{ x\in \F_{2^n} : (\forall \alpha \in S^{-\frac{1}{2^i+1}}) \Tr(\alpha x) = 0 \}$$
        where 
        $$S^{-\frac{1}{2^i+1}} = \{ 0, \mu^{-\frac{1}{2^i+1}}, (\xi\mu)^{-\frac{1}{2^i+1}}, ((1+\xi)\mu)^{-\frac{1}{2^i+1}}\}.$$
        Suppose towards a contradiction that $(\nu(b+b^{\frac{1}{2^i}}), \nu^{2^i+1}b)$ is a nonzero element belonging to $Y\cap Z$.
        Then $\nu^{2^i+1}b \in S\setminus\{ 0\} = \{\mu, \xi\mu, (1+\xi)\mu\}$. We will look at
        these three cases below. 

        First assume $\nu^{2^i+1}b = \mu$. Take 
        $\alpha = (\xi\mu)^{-\frac{1}{2^i+1}} \in S^{-\frac{1}{2^i+1}}$. 
        By Lemma~\ref{f8linear} we have $\alpha = \xi^{2^i}\mu^{-\frac{1}{2^i+1}}$. 
        Then by the condition defining $X$, we should have 
        $\Tr(\alpha\nu(b+b^{\frac{1}{2^i}})) = 0$. But
        \begin{align*}
            \Tr(\alpha\nu(b+b^{\frac{1}{2^i}}))
            &= \Tr(\xi^{2^i}\mu^{-\frac{1}{2^i+1}}\nu(\mu\nu^{-2^i-1} + \mu^{\frac{1}{2^i}}\nu^{-\frac{1}{2^i} -1}))\\
            &= \Tr(\xi^{2^i}(\mu^{\frac{2^i}{2^i+1}}\nu^{-2^i} + \mu^{\frac{1}{(2^i+1)2^i}}\nu^{-\frac{1}{2^i}}))\\
            &= \Tr(\xi^{2^i}\mu^\frac{2^i}{2^i+1}\nu^{-2^i}) + \Tr(\xi^{2^{3i}}\mu^\frac{2^{2i}}{(2^i+1)2^i}\nu^{-\frac{2^{2i}}{2^i}})\\
            &= \Tr(\xi^{2^i}\mu^\frac{2^i}{2^i+1}\nu^{-2^i}) + \Tr(\xi\mu^\frac{2^i}{2^i+1}\nu^{-2^i})\\
            &= \Tr((\xi + \xi^{2^i})\mu^\frac{2^i}{2^i+1}\nu^{-2^i}) = 1.
        \end{align*}

        Finally assume that $\nu^{2^i+1}b = z\mu$ and $z\in\{ \xi ,\xi + 1 \}$.
        Take $\alpha = \mu^{-\frac{1}{2^i+1}} \in S^{-\frac{1}{2^i+1}}$. 
        By the condition defining $X$, we should have 
        $\Tr(\alpha\nu(b+b^{\frac{1}{2^i}})) = 0$. But
        \begin{align*}
            \Tr(\alpha\nu(b+b^{\frac{1}{2^i}}))
            &= \Tr(\mu^{-\frac{1}{2^i+1}}\nu(z\mu\nu^{-2^i-1} + z^{\frac{1}{2^i}}\mu^{\frac{1}{2^i}}\nu^{-\frac{1}{2^i}-1}))\\
            &= \Tr(z\mu^{\frac{2^i}{2^i+1}}\nu^{-2^i} + z^{\frac{1}{2^i}}\mu^{\frac{1}{(2^i+1)2^i}}\nu^{-\frac{1}{2^i}})\\
            &= \Tr(z\mu^{\frac{2^i}{2^i+1}}\nu^{-2^i}) + \Tr(z^{\frac{2^{2i}}{2^i}}\mu^{\frac{2^{2i}}{(2^i+1)2^i}}\nu^{-\frac{2^{2i}}{2^i}})\\
            &= \Tr(z\mu^{\frac{2^i}{2^i+1}}\nu^{-2^i}) + \Tr(z^{2^i}\mu^{\frac{2^i}{2^i+1}}\nu^{-2^i})\\
            &= \Tr((\xi+\xi^{2^i})\mu^{\frac{2^i}{2^i+1}}\nu^{-2^i}) = 1.\\
        \end{align*}

        Since all cases are exhausted, the only element in the intersection is $(0,0)$ and we are done. 
    \end{proof}

In Propositions \ref{TI-hard-1} and \ref{TI-hard-2}
it is not difficult to see that if $\xi$ is drawn
uniformly at random from $\F_{2^3}\setminus \{0,1\}$
and $\mu,\nu$ are drawn
uniformly at random from $\F_{2^n}^*$,
then the probability that the trace condition
is satisfied is close to $1/2$. In other words,
we expect that a trivially intersecting pair
of WZ spaces is obtained for about every other set
of random parameters generated in this way.

\section{Applications}
    
    \subsection{Classifying EA classes of functions}
    
    Functions $f$ and $g$ mapping $\Fn$ to $\Fn$ are
    {\em extended affine equivalent (EA equivalent)} if there exist
    affine permutations $A_1,A_2$ of $\Fn$ and an affine function $A_3$
    such that $A_1(f(A_2(x)))+A_3(x)=g(x)$ 
    for all $x\in\Fn$. It is known that EA equivalent functions
    are also CCZ equivalent, but partitioning CCZ classes
    into EA classes is in general a hard problem.
    This problem was addressed by Canteaut and Perrin \cite{CP}
    by studying the structure of Walsh zeros of functions.
    WZ spaces play an important role in their investigations.
    
    To bring up a more specific example, in \cite[Lemma~12]{CP}
    it is stated that the CCZ class of $f(x)=x^3$ on $\F_{2^5}$
    contains three EA classes, and this is based on the classification
    of 64 WZ spaces that according to \cite{CP} were found
    experimentally. Here we can give a computer-free
    description of these spaces: 32 of them are
    obtained from Proposition~\ref{prop-traceFm} with $m=n=5$,  
    and the remaining 32 of them are obtained from Proposition~\ref{prop-comp}
    with $S=\mu\F_2$ where $\mu\in\F_{2^5}$.

    \subsection{Construction of new APN permutations}
    
    If we know two trivially intersecting WZ spaces
    for an APN function $f:\Fn\rightarrow\Fn$,
    then Proposition~\ref{prop-wz-disj} allows us
    to construct an APN permutation $f'$ of $\Fn$.
    Then $f'$ is CCZ equivalent to $f$, but in general
    it need not be EA equivalent to it. 
    
    Just for illustration we present a simple numerical example.
    For $n=9$,
    the algebraic degree of $f(x)=x^3$ is~2,
    and the algebraic degree of its compositional inverse
    $g(x)=x^{1/3}$, which is also an APN permutation,
    is~5. By applying Proposition~\ref{prop-wz-disj}
    along with constructions of trivially intersecting
    WZ spaces provided
    in Section~\ref{sec-Ti-pairs} above, we found
    APN permutations of $\F_{2^9}$ of algebraic degrees
    $2$, $4$ and $5$. Since the algebraic degree is preserved
    by EA equivalence, the APN permutations of degree $4$
    are not EA equivalent to $f$ or $g$.
    
    The constructions in Section~\ref{sec-Ti-pairs} work
    in arbitrary odd dimensions and for all Gold APN functions. It will be interesting
    to investigate how many EA inequivalent APN permutations they provide.
    
\section*{Acknowledgement}
	
	We thank Claude Carlet for pointing to us
	the current form of Proposition~\ref{prop-subfield-compatible}
	which was previously stated in a more special form. 
	We also thank
	L\'eo Perrin for releasing an updated version of 
	the sboxU software package \cite{sboxU} that is more
	widely compatible with various computer operating systems.

\end{document}